\documentclass{article}
\usepackage{amsmath,amsthm,amssymb}
\usepackage{array, longtable}
\newtheorem{Theorem}{Theorem}[section]
\newtheorem{lem}[Theorem]{Lemma}
\newtheorem{Remark}[Theorem]{Remark}

\newtheorem{Corollary}[Theorem]{Corollary}
\newtheorem{Proposition}[Theorem]{Proposition}

\newtheorem{Notation}[Theorem]{Notation}
\numberwithin{equation}{section}

\begin{document}

\title{Self-dual cyclic codes over finite chain rings\footnote{
 {\small Email addresses: bocong\_chen@yahoo.com (B. Chen),
  lingsan@ntu.edu.sg (S. Ling),
  zghui2012@126.com (G. Zhang).}}}

\author{Bocong Chen$^1$, ~San Ling$^1$, ~Guanghui Zhang$^2$}

\date{\small
${}^1$Division of Mathematical Sciences, School of Physical \& Mathematical Sciences,
         Nanyang Technological University, Singapore 637616, Singapore\\
${}^2$School of Mathematical Sciences,
Luoyang Normal University,
Luoyang, Henan, 471022, China}

\maketitle
\begin{abstract} Let $R$ be a finite commutative chain ring with  unique
 maximal ideal $\langle \gamma\rangle$,  and let $n$ be a positive integer coprime with the characteristic of $R/\langle \gamma\rangle$.
In this paper, the algebraic structure  of  cyclic codes
  of length $n$ over  $R$ is investigated.
  Some  new  necessary and sufficient conditions  for the existence of nontrivial self-dual cyclic codes are
 provided. An enumeration formula for the self-dual cyclic codes  is also studied.


\end{abstract}

\section{Introduction}
The study of  codes over finite rings has  grown
tremendously  since the seminal work of Hammons {\it et al}. It is  shown in \cite{Hammons} that some of the
best nonlinear codes  over $\mathbb{F}_2$ can be viewed as
linear codes over  $\mathbb{Z}_4$.
It was pointed out in \cite{Wood99,Wood08}  that only finite Frobenius rings
are suitable for coding alphabets,
in the sense that several fundamental properties of codes
over finite fields still hold for codes over such rings.
This has  motivated numerous  authors to research on  codes over finite chain rings,
as  chain rings  are Frobenius rings with good algebraic structures.

On the other hand, the class of cyclic codes plays a very significant role in the theory of error-correcting codes.
One is that they can be efficiently encoded  using shift registers.
There is a lot of literature  about cyclic codes over finite chain rings (e.g., see \cite{Abualrub}, \cite{Dinh04}-\cite{Dougherty}, \cite{Kai1},
\cite{Pless}-\cite{Sloane}).

Generally, cyclic codes over finite chain rings can be divided into two classes:
simple-root cyclic codes, if the code lengths are coprime with the characteristic of the ring; otherwise, we have the so-called
repeated-root cyclic codes.
In this paper, we study simple-root cyclic codes over finite chain rings.

Pless and Qian \cite{Pless} showed that   cyclic codes of odd length $n$ over $\mathbb{Z}_4$ have
generators of an interesting  form:
$\langle fh, 2gh\rangle$, where $f,g,h\in \mathbb{Z}_4[X]$ satisfy $fgh=X^n-1$.
Pless,  Sol\'{e} and Qian in \cite{Pless2} considered  existence conditions for  nontrivial
self-dual cyclic codes of odd length over $\mathbb{Z}_4$.
Results of \cite{Pless,Pless2} were then extended to
simple-root cyclic codes over $\mathbb{Z}_{p^m}$ \cite{Kanwar}.
Following that line of research, Wan continued to consider simple-root cyclic codes over Galois rings \cite{Wan99}.
Extending the main results of  \cite{Kanwar} and \cite{Wan99},
Dinh and L\'{o}pez-Permouth in \cite{Dinh04} completely described  simple-root cyclic codes over  a finite commutative chain ring $R$.
Several necessary and sufficient
conditions for the existence of nontrivial self-dual cyclic codes were provided.

Let $R$ be a finite commutative chain ring with  unique
maximal ideal $\langle \gamma\rangle$. Then $\langle \gamma\rangle$ is nilpotent and we denote its nilpotency index by $t$.
Let $n$ be a positive integer coprime
with the characteristic of $\mathbb{F}_q=R/\langle \gamma\rangle$.
First,   we generalize the methods  of \cite{Kai2} to obtain the
algebraic structure  of  cyclic codes of length $n$ over $R$, which is different from that given in \cite{Dinh04}.
Using this structure, we show that self-dual cyclic codes of length $n$ over $R$ exist if and only if
$t$ is even. Some  new  necessary and sufficient  conditions  for the existence of nontrivial self-dual cyclic codes are
also derived.
We show that, when the nilpotency index $t$ is even,
the number of self-dual cyclic codes
is fully determined by $|\Delta_n|$,
the number of reciprocal polynomial pairs in the monic irreducible  factorization of $X^n-1$ over $\mathbb{F}_q$.
The counting problem for $|\Delta_n|$ naturally reduces to an equivalent question about counting $|\Omega_n|$,
the number of self-reciprocal monic irreducible factors of $X^n-1$ over $\mathbb{F}_q$.
Write $n=2^mn'$, where $n'$ is odd.
It is shown that,
the problem of determining  the value of $|\Omega_{2^mn'}|$ can be entirely  reduced to those of
computing $|\Omega_{n'}|$ and $|\overline\Omega_{n'}|$,
where  $|\overline\Omega_{n'}|$ denotes the number of self-reciprocal monic irreducible factors of $X^{n'}-1$ over $\mathbb{F}_{q^2}$.
In particular,   very explicit formulas for the value of $|\Omega_{n}|$ are obtained when $n$ has exactly two  prime divisors.

This paper is organized as follows. After presenting  preliminary concepts and results  in Section 2,
we obtain structure theorems for cyclic codes of length $n$ over $R$ in Section 3. In Section 4, we provide some results concerning
the structure, existence conditions and enumeration formula for self-dual cyclic codes.
In Section 5, we study an   enumeration formula for the number of self-dual cyclic codes.

\section{Preliminaries}
A finite commutative ring with identity  is called a {\it finite chain ring} if it is  local and its unique maximal ideal is principal.
Throughout this paper,  $R$ denotes a finite  chain ring.
Let $\gamma$ be a fixed generator of the unique maximal ideal of $R$,
and assume that $\langle \gamma\rangle=\{r\gamma\,|\, r\in R\}$ is the principal ideal of $R$ generated by $\gamma$.
Then $\langle \gamma\rangle$ is nilpotent and we denote its   nilpotency index by  $t$.
We set $\mathbb{F}_q=R/\langle \gamma\rangle$. Here $\mathbb{F}_q$ is the finite field with $q=p^{\alpha}$ elements,
where $p$ is the characteristic of $\mathbb{F}_q$.

The natural surjective ring homomorphism from $R$ onto $\mathbb{F}_q$ is given as follows:
$$
^{-}:~~R~\longrightarrow ~\mathbb{F}_q, ~~~r\mapsto~\bar r, ~~~~\hbox{for any $r\in R$}.
$$
The map ${\bf``^-"}$ can be extended to a ring homomorphism from $R[X]$ onto $\mathbb{F}_q[X]$ in an obvious way:
$$
~~R[X]~\longrightarrow ~\mathbb{F}_q[X],~~~\sum\limits_{i=0}^na_iX^i~~\mapsto~~\sum\limits_{i=0}^n\overline{a_i}X^i,~~~
\hbox{for any $a_0,a_1, \cdots, a_{n}$ in $R$},
$$
which is also denoted by $``^-"$ for simplicity.

Two polynomials $f_1(X), f_2(X)$ in $R[X]$ are called {\it coprime} if there exist polynomials $u_1(X),u_2(X)$ in $R[X]$ such that
$u_1(X)f_1(X)+u_2(X)f_2(X)=1$. The following  result is very useful (e.g.,  see \cite[Lemma 2.8]{Norton} or \cite[Lemma 14.19]{Wan}).

\begin{lem}\label{1}
Let $f_1(X), f_2(X)$ be two polynomials in $R[X]$. Then $f_1(X), f_2(X)$ are coprime in  $R[X]$ if and only if
$\overline{f_1(X)}, \overline{f_2(X)}$ are coprime in $\mathbb{F}_q[X]$.
\end{lem}

A polynomial $f(X)\in R[X]$ is
said to be {\it basic irreducible} if $\overline{f(X)}$ is irreducible in $\mathbb{F}_q[X]$.
A polynomial $f(X)\in R[X]$ is called {\it regular} if it is not a zero divisor.
Clearly,  monic polynomials are regular polynomials.
A polynomial over a field is called {\it square free} if it has no multiple irreducible divisors in its decomposition.

Hensel's Lemma \cite[Theorem XIII.4]{McDonald} plays a very significant role in
the study of finite chain rings as well as codes  over finite chain rings.
Using Hensel's Lemma, it is easy to get the next result given in \cite[Lemma 2.3]{Kai2}.
\begin{lem}\label{2}
Let $f$  be a monic polynomial over $R$ such that $\overline{f}$ is square free.
If $\overline{f}=g_1g_2\cdots g_s$ is the unique factorization into a product of pairwise
coprime monic irreducible polynomials in $\mathbb{F}_q[X]$, then there exists a unique family of pairwise
coprime monic basic irreducible polynomials $f_1, f_2,\cdots,f_s$ over $R$ such that
$f=f_1f_2\cdots f_s$ and $\overline{f_i}=g_i$ for $1\leq i\leq s$.
\end{lem}

In the rest of this section,  we recall some notations and
basic facts about  codes over rings.
Let $n$ be a positive integer.
A {\it code} $C$ of length $n$ over $R$ is a nonempty subset of $R^n$.
If,  in addition,   $C$  is  an $R$-submodule of $R^n$, then $C$ is called a {\em linear code}.
A linear code $C$ of length $n$ over $R$ is called
{\em cyclic} if $(c_{n-1}, c_0,\cdots,c_{n-2})\in C$
for every $(c_{0}, c_1,\cdots,c_{n-1})\in C$.

Each codeword $c=(c_0, c_1,\cdots, c_{n-1})$ is customarily identified with its polynomial representation
$c(X)=c_0+c_1X+\cdots+c_{n-1}X^{n-1}$. In this way, any cyclic code of length $n$ over $R$ is identified with exactly one ideal of
the quotient algebra $R[X]/\langle X^n-1\rangle$.

For any cyclic  code $C$ of length $n$ over $R$, the {\it dual code of $C$}
is defined as $C^\perp=\{u\in R^n\,|\,u\cdot v=0, ~\mbox{for any $v\in C$}\}$,
where $u\cdot v$ denotes the standard Euclidean inner product of $u$ and $v$ in $R^n$.
The code $C$ is said to be {\em self-orthogonal}
if $C\subseteq C^\perp$,  and {\it self-dual} if $C=C^\perp$.
It turns out that   the dual of a cyclic code is again a cyclic code.
The following result is well known (e.g.,  see \cite[Proposition 2.3]{Dinh10}).

\begin{lem}
Let  $R$ be a finite chain ring. Then the number of codewords in any
linear code $C$ of length $n$ over $R$ satisfies $|C|\cdot|C^\perp|=|R|^n$.
\end{lem}

\section{Structure of cyclic codes over finite chain rings}
Starting from this section till the end of this paper, we always assume that $n$ is a positive integer coprime with the characteristic of $\mathbb{F}_q$.
We set $\mathcal{R}_n:=R[X]/\langle X^{n}-1\rangle$. Recall that $R$ is a finite chain ring with maximal
ideal $\langle \gamma\rangle$,  and that $\mathbb{F}_q=R/\langle \gamma\rangle$ is the residue field of order $q=p^\alpha$.
We  adopt the following notations.

\begin{Notation}
Let $r_0$ be an element in  $R$ such that $r_0=1+\mu\gamma$, where $\mu$ is a unit in $R$.
It follows that  $\overline {X^n-r_0}=X^n-1$ in $\mathbb{F}_q[X]$.
Since  $q$ is coprime with $n$,
the irreducible factors of $X^n-1$ in $\mathbb{F}_q[X]$ can be described by the $q$-cyclotomic cosets.
Let $I$ be a fixed complete set of  representatives of  all $q$-cyclotomic cosets modulo $n.$
Then the polynomial $X^n-1$ factors uniquely into pairwise coprime monic  irreducible
polynomials in $\mathbb{F}_q[X]$ as $X^n-1=\prod\limits_{i\in I}h_i$ (e.g.,  see \cite[Theorem 4.1.1]{Huffman}).
Using Lemma~\ref{2},
$X^n-r_0$ has a unique decomposition as a product	$\prod\limits_{i\in I}f_i$ of
pairwise coprime monic basic irreducible
polynomials in $R[X]$ with $\bar f_i=h_i$ for each $i \in I$.
Using Lemma~\ref{2} again,
$X^n-1\in R[X]$ also has a unique decomposition as a
product $\prod\limits_{i \in I}g_i$ of pairwise coprime
 monic basic irreducible polynomials in $R[X]$ with  $\bar g_i=h_i$ for each $i \in I$.

For any commutative ring $S$, elements $s_1, s_2\in S$ are said to be associates
if there is a unit $\lambda\in S$ with  $\lambda s_1=s_2$.
For a monic  polynomial $h(X)$ of degree $k$ in $S[X]$ with $h(0)$ being a unit in $S$, its  reciprocal
polynomial $h(0)^{-1}X^kh(X^{-1})$  is denoted by $h(X)^*$.
Note that $h(X)^*$ is a monic polynomial over $S$.
Following \cite{Jia},
if $h(X)=h(X)^*$,
then $h(X)$ is said to be self-reciprocal over $S$; otherwise, we say that
$h(X)$ and  $h(X)^*$ form a reciprocal polynomial pair.
For a polynomial $f(X)\in S[X]$  with leading coefficient
$a_n$ being a unit of $S$,  let $\hat f(X)=a_n^{-1}f(X)$, which is a monic  polynomial over
$S$.
\end{Notation}

\begin{lem} {\rm (\cite[Lemma 3.1]{Dinh04})}\label{Dinhlem}
Let $R$ be a finite chain ring with  maximal
ideal $\langle \gamma\rangle$, and let $t$ be the nilpotency index of $\gamma$. If $g$ is a  basic
irreducible polynomial of the ring $R[X]$, then $R[X]/\langle g\rangle$ is also a finite chain ring,
whose maximal ideal is generated by $\gamma+\langle g\rangle$. The nilpotency index of $\gamma+\langle g\rangle$ is equal to $t$.
\end{lem}

The next lemma shows that the element $f_i+\langle g_i\rangle$
generates the maximal ideal of $R[X]/\langle g_i\rangle$.

\begin{lem}\label{lemma1}
Let the notation be as above. Then,   for each $i \in I$,  we have that
$\big\langle f_i^{k_i}+\langle g_i\rangle\big\rangle=\big\langle \gamma^{k_i}+\langle g_i\rangle\big\rangle$
as ideals in $R[X]/\langle g_i\rangle$,
where $k_i$ is any nonnegative integer.
\end{lem}
\begin{proof}
It suffices to show that $f_i+\langle g_i\rangle$ and $\gamma+\langle g_i\rangle$ are associates in $R[X]/\langle g_i\rangle$.
By Lemma~\ref{1},  $\widetilde{ f}_i=\frac{X^n-r_0}{f_i}$ is coprime with $g_i$ in $R[X]$, since $\overline{\widetilde{f}_i}$ is coprime with
$\overline{g_i}$ in $\mathbb{F}_q[X]$. Thus there exist elements $u_i, v_{i}$ in $R[X]$ such that
\begin{equation*}
u_i\widetilde{f}_i+v_ig_i=1.
\end{equation*}
Then in $R[X]/\langle g_i\rangle$,
\begin{equation}\label{one}
\big(u_i+\langle g_i\rangle\big)\big(\widetilde{f}_i+\langle g_i\rangle\big)=1+\langle g_i\rangle,
\end{equation}
which means $u_i+\langle g_i\rangle$ is a unit in $R[X]/\langle g_i\rangle$.
Multiplying by $f_i+\langle g_i\rangle$ on both sides of (\ref{one}) gives
$u_i(X^n-r_0)+\langle g_i\rangle=f_i+\langle g_i\rangle$.
Then
\begin{equation*}
(u_i+\langle g_i\rangle)^{-1}(f_i+\langle g_i\rangle)=
X^n-r_0+\langle g_i\rangle =X^n-1+1-r_0+\langle g_i\rangle =1-r_0+\langle g_i\rangle
=-\gamma\mu+\langle g_i\rangle.
\end{equation*}
We are done.
\end{proof}

By the Chinese Remainder Theorem, we have the following $R$-algebra isomorphism:
\begin{equation}\label{CRT1}
\varphi:~\mathcal{R}_n=R[X]/\langle X^n-1\rangle~
\buildrel{\cong}\over{\longrightarrow}~\bigoplus\limits_{i\in I}\frac{R[X]}{\langle g_i\rangle}.
\end{equation}
Making use of (\ref{CRT1}), Dinh and L\'{o}pez-Permouth obtained the structure of cyclic codes of length $n$ over $R$.
It was shown that any ideal in $\mathcal{R}_n$ is a sum of ideals of the form $\big\langle \gamma^j\hat g_i+\langle X^n-1\rangle\big\rangle$
(see \cite[Theorem 3.2]{Dinh04}).
In the light of Lemma~\ref{lemma1}, we have another characterization  of  cyclic codes of length $n$ over $R$
with polynomial generators in terms of   $f_i$, $i \in I$.

\begin{Theorem}\label{Theorem1}
Let the notation be the same as before.  Let $C$ be a cyclic code of length $n$ over $R$.
Then $C\cong\bigoplus\limits_{i\in I}\big\langle \gamma^{k_i}+\langle g_i\rangle\big\rangle$
under the map given by (\ref{CRT1}) if and only if
$C=\big\langle\prod\limits_{i\in I}f_{i}^{k_i}\big\rangle$, where
$\big\langle \gamma^{k_i}+\langle g_i\rangle\big\rangle$ is an ideal of $R[X]/\langle g_i\rangle$ with $0\leq k_i\leq t$.
In this case,  $|C|=q^{\sum\limits_{i\in I}(t-k_i)\deg f_i}$.
Moreover, for any ideal $C$ in $\mathcal{R}_n$, there exists a unique sequence
$(k_i)_{i\in I}$,  with $0\leq k_i\leq t$,  such that $C=\big\langle\prod\limits_{i\in I}f_{i}^{k_i}\big\rangle$.
\end{Theorem}
\begin{proof}
From Lemma~\ref{lemma1},
$$
\bigoplus\limits_{i\in I}\big\langle \gamma^{k_i}+\langle g_i\rangle\big\rangle=
\bigoplus\limits_{i\in I}\big\langle f_{i}^{k_i}+\langle g_i\rangle\big\rangle.
$$
It is readily seen that,  under the map $\varphi$ given by (\ref{CRT1}),
$$
\varphi\Big(\big\langle\prod\limits_{j\in I}f_{j}^{k_j}\big\rangle\Big)=
\bigoplus\limits_{i\in I}\big\langle \prod\limits_{j\in I}f_{j}^{k_j}+\langle g_i\rangle\big\rangle=
\bigoplus\limits_{i\in I}\big\langle f_{i}^{k_i}+\langle g_i\rangle\big\rangle.
$$
The last equality holds because $\prod\limits_{j\in I\setminus\{i\}}f_{j}^{k_j}+\langle g_i\rangle$
is a unit in $R[X]/\langle g_i\rangle$.
We have also shown that,  for any cyclic code $C$ of length $n$ over $R$, there  exists a  sequence
$(k_i)_{i\in I}$ with $0\leq k_i\leq t$ such that $C=\big\langle\prod\limits_{i\in I}f_{i}^{k_i}\big\rangle$.
Uniqueness can be proved as follows: if
$\big\langle\prod\limits_{i\in I}f_{i}^{k_i}\big\rangle=\big\langle\prod\limits_{i\in I}f_{i}^{k_{i}'}\big\rangle$
with $0\leq k_i, k_{i}'\leq t$ for all $i\in I$,
then
$$
\bigoplus\limits_{i\in I}\big\langle \gamma^{k_i}+\langle g_i\rangle\big\rangle=
\varphi\Big(\big\langle\prod\limits_{i\in I}f_{i}^{k_i}\big\rangle\Big)=
\varphi\Big(\big\langle\prod\limits_{i\in I}f_{i}^{k_{i}'}\big\rangle\Big)=
\bigoplus\limits_{i\in I}\big\langle \gamma^{k_{i}'}+\langle g_i\rangle\big\rangle.
$$
This forces $k_i=k_{i}'$ for all $i\in I$.

To  complete the proof, we have
$$
|C|=\prod\limits_{i\in I}|\big\langle \gamma^{k_i}+\langle g_i\rangle\big\rangle|
=\prod\limits_{i\in I}q^{(t-k_i)\deg g_i}=
q^{\sum\limits_{i\in I}(t-k_i)\deg f_i}.
$$

\end{proof}

\begin{Remark}
The pairwise coprime  monic basic irreducible
factors of $X^n-r_0$ in $R[X]$  can be easily derived from the
pairwise coprime monic basic irreducible
factors of $X^n-1$ in $R[X]$.
To see this, observe that $r_0=1+\gamma\mu$ is an element in the Sylow $p$-subgroup of $R^*$,
where $R^*$ stands for the unit group of $R$.
Let $P$ be the Sylow $p$-subgroup of $R^*$.
Since $\gcd(n,p)=1$, then $\theta:~P~\mapsto~P$, defined by $\theta(y)=y^n$,
is  actually  an  automorphism.
Thus, we can find a unique element $\delta$ of $P$
such that $\delta^nr_0=1$.
Therefore,  if we already have the
pairwise coprime monic basic irreducible
factorization of $X^n-1$ in $R[X]$:
$$
X^n-1=\prod\limits_{i\in I}g_i(X),
$$
then  substitute $X$ for $\delta X$ to obtain
the pairwise coprime
monic basic irreducible
factorization of $X^n-r_0$ in $R[X]$:
$$
X^n-r_0=\prod\limits_{i\in I}\Big(\delta^{-\deg g_i}g_i(\delta X)\Big).
$$
\end{Remark}

\section{Dual cyclic codes}

We can also characterize the dual code of $C$ in terms of the polynomials $f_i$, $i\in I$.
Recall that $X^n-1=\prod_{i\in I}h_i$ gives the monic irreducible factorization
of $X^n-1$ in $\mathbb{F}_q[X]$, and that
$X^n-r_0=\prod_{i\in I}f_i$ is the
pairwise coprime monic basic irreducible
factorization in $R[X]$ with $\bar f_i=h_i$ for  subscripts in this range.
Observe that $h_i^*$ is also a monic divisor of $X^n-1$ in $\mathbb{F}_q[X]$.
Thus, for each $i\in I$,  there exists a unique $i'\in I$
such that $h_{i'}=h_i^*$. This implies that $'$ is a bijection
from $I$ onto $I$, which   satisfies $(i')'=i$ for all $i\in I$.

\begin{lem}\label{lemma4}
With respect to the above notation, let $C=\big\langle\prod\limits_{i\in I}(f_{i}^*)^{k_i}\big\rangle$
be a cyclic code of length $n$ over $R$ with  $0\leq k_i\leq t$.  We then have
$C=\big\langle\prod\limits_{i\in I}f_{i'}^{k_i}\big\rangle$.
\end{lem}
\begin{proof}
We know that $\overline{f_i^*}=\overline{f_i}^*=h_i^*=h_{i'}=\overline{f_{i'}}$, which implies that
there exists $q_i\in R[X]$  such that $f_{i'}=f_i^*+\gamma q_i$.
Now in $\mathcal{R}_n$,
$$
r_0\prod\limits_{i\in I}f_i^*=r_0(X^{n}-r_0^{-1})=1+\mu \gamma-1=\mu\gamma
$$
and
$$
\prod_{i\in I}f_{i'}=X^n-r_0=1-1-\mu\gamma=-\mu\gamma.
$$
It follows that
\begin{equation*}
\prod\limits_{i\in I}f_{i'}^{k_i}
=\prod\limits_{i\in I}(f_i^*+\gamma q_i)^{k_i}
=\prod\limits_{i\in I}(f_i^*+\mu^{-1}r_0q_i\prod\limits_{j\in I}f_j^*)^{k_i},
\end{equation*}
\begin{equation*}
\prod\limits_{i\in I}(f_i^*)^{k_i}
=\prod\limits_{i\in I}(f_{i'}-\gamma q_i)^{k_i}
=\prod\limits_{i\in I}(f_{i'}+\mu^{-1}q_i\prod\limits_{j\in I}f_{j'})^{k_i}.
\end{equation*}
Clearly $\prod\limits_{i\in I}(f_i^*)^{k_i}$ is a divisor of $\prod\limits_{i\in I}f_{i'}^{k_i}$ and vice versa.
We have obtained the desired result.
\end{proof}

\begin{lem}
Let $C=\big\langle\prod\limits_{i\in I}f_{i}^{k_i}\big\rangle$
be a  cyclic code of length $n$ over $R$, where the polynomials
$f_i$ are the pairwise coprime monic basic irreducible factors of $X^n-r_0$ in $R[X]$
and  $0\leq k_i\leq t$ for each $i\in I$.
Then $C^\perp=\big\langle\prod\limits_{i\in I}f_{i'}^{t-k_i}\big\rangle$
and $|C^\perp|=q^{\sum\limits_{i\in I}k_i\deg f_i}$.
\end{lem}
\begin{proof}
By Theorem~\ref{Theorem1},
$$
|C^\perp|=\frac{|R|^n}{|C|}=\frac{q^{nt}}{q^{\sum\limits_{i\in I}(t-k_i)\deg f_i}}=q^{\sum\limits_{i\in I}k_i\deg f_i}
=q^{\sum\limits_{i\in I}k_{i}\deg f_{i'}}
=|\big\langle\prod\limits_{i\in I}f_{i'}^{t-k_i}\big\rangle|.
$$
The fourth  equality holds because $\deg f_{i'}=\deg f_i$.
From Lemma~\ref{lemma4} $\big\langle\prod\limits_{i\in I}(f_{i}^*)^{t-k_i}\big\rangle=\big\langle\prod\limits_{i\in I}f_{i'}^{t-k_i}\big\rangle$,
it remains  to prove that $\big\langle\prod\limits_{i\in I}(f_{i}^*)^{t-k_i}\big\rangle\subseteq C^\perp$.
Following \cite[Proposition 2.12]{Dinh04}, it suffices to show that
$\prod\limits_{i\in I}f_{i}^{k_i}\cdot\big(\prod\limits_{i\in I}(f_{i}^*)^{t-k_i}\big)^*=0$ in
$\mathcal{R}_n$. Indeed,
$$
\prod\limits_{i\in I}f_{i}^{k_i}\cdot\big(\prod\limits_{i\in I}(f_{i}^*)^{t-k_i}\big)^*=\delta\prod\limits_{i\in I}f_{i}^{t}
=\delta(X^n-1-\mu\gamma)^t=0,
$$
where $\delta$ is a suitable unit of $\mathcal{R}_n$.
\end{proof}

We now produce a criterion to determine whether or not a given  cyclic code of length $n$ over $R$ is self-dual.

\begin{Theorem}\label{Theorem2}
Let $C=\big\langle\prod\limits_{i\in I}f_{i}^{k_i}\big\rangle$ be a cyclic code of length $n$ over $R$, where
$f_i$ are the  pairwise coprime monic basic irreducible factors of $X^n-r_0$ in $R[X]$ and $0\leq k_i\leq t$.
Then $C$ is self-dual if and only if $k_i+k_{i'}=t$ for all $i\in I$.
\end{Theorem}
\begin{proof}
Recall that $'$ is a bijection
from $I$ onto $I$, which   satisfies $(i')'=i$ for all $i\in I$.
Then
$$
C^\perp=
\big\langle\prod\limits_{i\in I}f_{i'}^{t-k_i}\big\rangle
=\big\langle\prod\limits_{i\in I}f_{(i')'}^{t-k_{i'}}\big\rangle
=\big\langle\prod\limits_{i\in I}f_{i}^{t-k_{i'}}\big\rangle.
$$
Comparing with $C=\big\langle\prod\limits_{i\in I}f_{i}^{k_i}\big\rangle$,
it follows that $C=C^\perp$ if and only if $k_i+k_{i'}=t$ for all $i\in I$.
\end{proof}

From the criterion above,  we are led to a simple condition for
the existence of
self-dual cyclic codes over finite chain rings.

\begin{Theorem}\label{corollary}
Let the notation be the same as before. Then there exists a
self-dual cyclic code of length $n$ over $R$ if and only if $t$, the nilpotency index of $R$, is even.
\end{Theorem}
\begin{proof}
If $t$ is even,  then
$\langle\prod\limits_{i\in I}f_{i}^{\frac{t}{2}}\rangle$ is a self-dual cyclic code
of length $n$ over $R$.

Conversely, assume that there exists a  self-dual cyclic code
$C=\langle\prod\limits_{i\in I}f_{i}^{k_i}\rangle$ of length $n$ over $R$.
From Theorem~\ref{Theorem2}, $k_i+k_{i'}=t$ for all $i\in I$.
In particular,   0 is always an element in $I$ with $0'=0$.
It follows that $2k_0=t$, which gives the desired result.
\end{proof}

Recall that $I$ is a fixed complete set of  representatives of  all $q$-cyclotomic cosets modulo $n.$
Let $\Omega_n$ and $\Delta_n$ be the sets $\Omega_n=\{i\in I\,|\,i'=i\}$ and
$\Delta_n=\{i\in I\,|\,i'\neq i\}=\{i_1,i_1',\cdots,i_s,i_s'\}$ respectively.
Clearly $I$ is the disjoint union of $\Omega_n$ and $\Delta_n$, $I=\Omega_n\,\cup\,\Delta_n$.
It follows that $X^n-r_0=\prod\limits_{i\in \Omega_n}f_i\cdot\prod\limits_{j=1}^sf_{i_j}f_{i_j'}$.
Similar to \cite[Theorem 2]{Jia}
and \cite[Corollary 1]{Jia}, we can characterize all self-dual cyclic codes according to
the sets $\Omega_n$ and $\Delta_n$.
\begin{Corollary}\label{number}
With respect to the above notation, assume that $t$ is even.  We then have that
$C$ is a self-dual cyclic code of length $n$ over $R$ if and only if
$C$ can be expressed as the form
$\langle \prod\limits_{i\in \Omega_n}f_i^{\frac{t}{2}}\cdot \prod\limits_{j=1}^sf_{i_j}^{k_j}f_{i_j'}^{t-k_j}\rangle$,
where $k_j$ are  integers with $0\leq k_j\leq t$.
In particular,  there are exactly $(t+1)^s=(t+1)^{\frac{|\Delta_n|}{2}}$ self-dual cyclic codes of length $n$ over $R$.
\end{Corollary}

When the nilpotency index $t$ is even,
the self-dual cyclic code $\langle\prod_{i\in I}f_{i}^{\frac{t}{2}}\rangle$ is called  {\it trivial self-dual code}.
In order to investigate the existence conditions for nontrivial self-dual cyclic codes,
we need the following observation.

\begin{lem}\label{lemma13}
Let $g\in R[X]$ be a monic basic irreducible factor
of $X^n-1$. Let $h\in \mathbb{F}_q[X]$ be the image of $g$ under the surjective ring homomorphism
$``^-"$ from $R[X]$ onto $\mathbb{F}_q[X]$, namely $\bar g=h$. Then $g$ and $g^*$ are  associates in $R[X]$
if and only if $h$ and $h^*$ are associates in $\mathbb{F}_q[X]$.
\end{lem}
\begin{proof}
Obviously, if $g$ and $g^*$ are associates in $R[X]$ then $h$ and $h^*$ are associates in $\mathbb{F}_q[X]$.

Conversely, assume that $h$ and $h^*$ are associates in $\mathbb{F}_q[X]$.
Suppose otherwise that $g$ and $g^*$ are not associates in $R[X]$.
Consider the homomorphism $\rho$ as given in the proof of Lemma~\ref{Dinhlem},
$$
\rho:~~R[X]/\langle g\rangle~\longrightarrow~\mathbb{F}_q[X]/\langle h\rangle,
~~~\sum\limits_{j=0}^{n}a_jX^j+\langle g\rangle\mapsto~\sum\limits_{j=0}^{n}\overline{a_j}X^j+\langle h\rangle,
~~~\hbox{for any $a_0,a_1, \cdots,a_{n}$ in $R$}.
$$
On the one hand,   $g^*$ is coprime with $g$  in $R[X]$.
This implies that $g^*+\langle g\rangle$ is a unit in $R[X]/\langle g\rangle$, and so is
$\rho\big(g^*+\langle g\rangle\big)$ in  $\mathbb{F}_q[X]/\langle h\rangle$.
On the other hand, $\rho\big(g^*+\langle g\rangle\big)=h^*+\langle h\rangle=0$ in $\mathbb{F}_q[X]/\langle h\rangle$.
This is a contradiction.
\end{proof}

\begin{Remark}
It follows from Theorem~\ref{corollary} and Corollary~\ref{number} that  nontrivial self-dual cyclic codes of length $n$ over $R$ exist if and only if
$t$ is even and $|\Delta_n|>0$.
Clearly,  the  condition $|\Delta_n|>0$ holds if and only if there exists a monic irreducible factor $h\in \mathbb{F}_q[X]$ of $X^n-1$ such that
$h$ and $h^*$ are not associates.
Thanks to Lemma~\ref{lemma13},
nontrivial self-dual cyclic codes of length $n$ over $R$ exist
if and only if there exists a monic basic irreducible factor $g\in R[X]$ of $X^n-1$ such that
$g$ and $g^*$ are not associates.
In conclusion,  we have the following result.
\end{Remark}

\begin{Theorem}\label{nontrivial}
Assume that the nilpotency index $t$ is even. The following five statements are equivalent to one another:

{\rm (i)}~Nontrivial self-dual cyclic codes of length $n$ over $R$ exist.

{\rm (ii)}~The cardinality of the set $\Delta_n$ is  nonzero, i.e., $|\Delta_n|>0$.

{\rm (iii)}~$q^{i}\not\equiv-1~(\bmod~n)$ for all positive integer $i$,
where $q$ is the order of the residue field $\mathbb{F}_q=R/\langle \gamma\rangle$.

{\rm (iv)}~There exists a monic irreducible factor $h\in \mathbb{F}_q[X]$ of $X^n-1$ such that
$h$ and $h^*$ are not associates.

{\rm (v)}~There exists a monic basic irreducible factor $g\in R[X]$ of $X^n-1$ such that
$g$ and $g^*$ are not associates.
\end{Theorem}

Note that the equivalence  of {\rm (i)}, {\rm (iii)} and {\rm (v)} appeared  previously in
\cite[Theorem 4.3]{Dinh04} and \cite[Theorem 4.4]{Dinh04}.

\section{Enumeration of  self-dual cyclic codes }
In this section,  we study
an enumeration formula for  self-dual cyclic codes of length $n$ over $R$.
It follows from Corollary~\ref{number} that, if the nilpotency index $t$ is even,  this number is fully determined by $|\Delta_n|$,
the number of reciprocal polynomial pairs in the monic irreducible  factorization of $X^n-1$ over $\mathbb{F}_q$.
Recall that the value $|\Delta_n|+|\Omega_n|$ is exactly equal to  the number of all monic irreducible
factors of $X^n-1$ over $\mathbb{F}_q$,
where $|\Omega_n|$ is the cardinality of all self-reciprocal monic irreducible factors of
$X^n-1$ over $\mathbb{F}_q$.
Meanwhile, one knows that the number of monic irreducible
factors of $X^n-1$ over $\mathbb{F}_q$ can be explicitly given by
$\sum_{d\mid n}\frac{\phi(d)}{{\rm ord}_d(q)}$,
where $\phi$ is   Euler's function.
Thus,  the counting problem for $|\Delta_n|$ naturally reduces to the equivalent question
of determining  the size of $\Omega_n$.

\subsection{An enumeration formula for $|\Omega_{2^m}|$}

We first consider the case when the code length $n$ is a power of 2, $n=2^m$.
The value
$|\Omega_{2^m}|$  can be easily determined.
In fact,  the irreducible factorization of $X^{2^m}-1$ over $\mathbb{F}_q$ has been given explicitly
(e.g.,  see \cite{li} or \cite[Theorem 3.1]{Chen2} for the case $q\equiv1~(\bmod~4)$,
and see \cite[Corollary 4]{Blake} or \cite[Lemma 2.2]{Chen} for the case $q\equiv-1~(\bmod~4)$).
For  convenience,  we reproduce these results below.

\begin{lem}
Assume that $q\equiv1~(\bmod~4)$.
Write $q-1=2^vc$ with $\gcd(2,c)=1$
and $v\geq2$.
Let $\eta$ be a primitive $2^v$th root of unity in $\mathbb{F}_q$.
Then
\begin{equation*}
X^{2^m}-1=\left\{
                   \begin{array}{ll}
                     \prod\limits_{k=0}^{2^v-1}\big(X-\eta^k\big)\cdot
\prod\limits_{j=1}^{m-v}\prod\limits_{i=1 \atop{2\,\nmid \,i}}^{2^v-1}
 \big(X^{2^j}-\eta^i\big), & \hbox{if $m>v$;} \\
                     \prod\limits_{k=0}^{2^m-1}\big(X-\delta^k\big), & \hbox{if $m\leq v$,}
                   \end{array}
                 \right.
\end{equation*}
where $\delta$ is a primitive $2^m$th root of unity in $\mathbb{F}_q$ for
$m\leq v$. All the factors on the right hand side of the equation
above are irreducible over $\mathbb{F}_{q}.$
\end{lem}

Next is the case $q\equiv-1~(\bmod~4)$.
Note that $4\,|\,(q+1)$ in this case,
hence there is a unique integer $a\ge 2$ such that $2^a\Vert(q+1)$,
where  the notation
$2^a\Vert (q+1)$  means   $2^a\,|\,(q+1)$ but $2^{a+1}\nmid (q+1)$.

\begin{lem}\label{irr-trinomial} Assume that
$q\equiv-1~(\bmod~4)$.
Set $H_1=\{0\}$; recursively define
$$\textstyle H_i=\left\{\pm(\frac{h+1}{2})^\frac{q+1}{4}\,|\,h\in H_{i-1}\right\},$$
for $i=2,3, \cdots, a-1$; and set
$$\textstyle H_a=\left\{\pm(\frac{h-1}{2})^\frac{q+1}{4}\,|\,h\in H_{a-1}\right\}.$$
Then
for $1\leq i\leq a$, $H_i$ has    cardinality $2^{i-1}$.
The
irreducible factorization of $X^{2^m}-1$ over $\mathbb{F}_q$ is given as follows:

If $1\leq m\leq a$, then
\begin{equation}\label{first}
X^{2^m}-1=(X-1)(X+1)\prod\limits_{i=1}^{m-1}\prod\limits_{h\in H_{i}}(X^2-2hX+1);
\end{equation}
if $m\geq a+1$, then
\begin{equation}\label{second}
X^{2^m}-1=(X-1)(X+1)\prod\limits_{h\in H_{i}, ~\atop{1\leq i\leq(a-1)}}(X^2-2hX+1)
\prod\limits_{h\in H_{a},~\atop{0\leq k\leq(m-a-1)}}(X^{2^{k+1}}-2hX^{2^{k}}-1).
\end{equation}
\end{lem}

The above two lemmas combine to give the following result.
\begin{Proposition}
The number
of self-reciprocal monic irreducible factors of $X^{2^m}-1$ over $\mathbb{F}_q$
is explicibly given by
\begin{equation}\label{omega1}
|\Omega_{2^m}|=\left\{
                   \begin{array}{lll}
                     1, & \hbox{if $m=0$;} \\
2, & \hbox{if $m=1$ or $m\geq2$ and $4\mid(q-1)$;}\\
2^{\min\{m,a\}-1}+1 & \hbox{if  $m\geq2$ and $4\nmid(q-1)$.}
                   \end{array}
                 \right.
\end{equation}
\end{Proposition}

\subsection{A reduction formula for $|\Omega_{2^mn'}|$ }

We turn our attention to the more general case.
Let  $n=2^mn'$ with
$n'=p_1^{r_1}p_2^{r_2}\cdots p_{k}^{r_k}$, where
$p_j$ are distinct odd primes and $r_j$
are positive integers for $1\leq j\leq k$.
Our major goal  is to show that,
the problem for determining  the value of $|\Omega_{2^mn'}|$ can be entirely  reduced to
computing $|\Omega_{n'}|$ and $|\overline\Omega_{n'}|$,
where $|\overline\Omega_{n'}|$ denotes the number of self-reciprocal monic irreducible factors of $X^{n'}-1$ over $\mathbb{F}_{q^2}$.

\begin{lem}\label{pairs}
Let $\ell$ be an odd prime integer coprime with $q$,  and let $s$  be a positive integer. Then
$|\Omega_{\ell^s}|=1$ if and only if
$q$ has odd order in the  multiplicative group of  integers modulo $\ell$, i.e.,
$2\nmid{\rm ord}_{\ell}(q)$.
\end{lem}
\begin{proof}
Note that ${\rm ord}_{\ell}(q)$ is odd if and only if ${\rm ord}_{\ell^s}(q)$ is odd.
Indeed,  if ${\rm ord}_{\ell}(q)=e$ is odd, then there is an integer $k$ such that
$q^e\equiv1+\ell k~(\bmod~\ell^s)$. Now the desired result follows from the fact that
the natural surjective homomorphism $\pi:~\mathbb{Z}_{\ell^s}^*~\mapsto~\mathbb{Z}_{\ell}^*$
with $|Ker\pi|$ being odd,
and $[1+\ell k]_{\ell^s}\in Ker\pi$.

Let $C_i=\{i\cdot q^j\pmod{\ell^s}\,|\,j=0,1,\ldots\}$ be the
$q$-cyclotomic coset modulo $\ell^s$ containing $i$.
Equivalently, we need to prove that
$C_{i}\neq C_{-i}$ for any integer  $i\not\equiv0\pmod{\ell^s}$ if and only if ${\rm ord}_{\ell}(q)$ is odd.

We first assume that $C_{i}\neq C_{-i}$ for any integer  $i\not\equiv0\pmod{\ell^s}$.
Recall that $\mathbb{Z}_{\ell^s}^*=\{[k]_{\ell^s}\,|\, \gcd(k, \ell)=1 \}$ is a cyclic group, which implies that $[-1]_{\ell^s}$
is the unique element of $\mathbb{Z}_{\ell^s}^*$ with order 2.
If $f={\rm ord}_{\ell^s}(q)$ is even,
then $q^{\frac{f}{2}}\equiv-1\pmod{\ell^s}$.
This is a contradiction, since we would obtain $C_1=C_{-1}$.

Conversely, assume that $f={\rm ord}_{\ell}(q)$ is odd.
Suppose otherwise that there exists an integer $i_0$ with
$i_0\not\equiv0\pmod{\ell^s}$ satisfying $C_{i_0}=C_{-i_0}$.
That is to say, an integer $j$ can be found so that
$q^ji_0\equiv-i_0\pmod{\ell^s}$.
We write $i_0=\ell^{s_0}a_0$ with $\gcd(\ell, a_0)=1$. Clearly
$s>s_0$. We then have $q^ja_0\equiv-a_0\pmod{\ell^{s-s_0}}$.
This leads to $q^{fj}a_0^f\equiv-a_0^f\pmod{\ell^{s-s_0}}$,
which implies that $q^{fj}a_0^f\equiv-a_0^f\pmod{\ell}$.
It follows that $a_0^f\equiv-a_0^f\pmod{\ell}$, and thus $\ell\mid a_0$. This is a contradiction.
\end{proof}

\begin{Remark}\label{even}
From the proof of Lemma~\ref{pairs}, one  can  easily deduce that
all the monic irreducible factors of $X^{\ell^s}-1$  over $\mathbb{F}_q$ are self-reciprocal
if and only if ${\rm ord}_{\ell}(q)$ is even.

At this point, we point out that, for any odd prime $\ell$ coprime with $q$, the value
$|\Omega_{\ell^s}|$ can be determined easily. Indeed, if ${\rm ord}_{\ell}(q)$ is odd, then $|\Omega_{\ell^s}|=1$;
otherwise $|\Omega_{\ell^s}|=\sum_{d=0}^{s}\frac{\phi(\ell^d)}{{\rm ord}_{\ell^d}(q)}$,
the number of all monic irreducible factors of $X^{\ell^s}-1$ over $\mathbb{F}_q$.
\end{Remark}
For computing the value of $|\Omega_n|$,
the following lemma asserts  that the odd prime divisor $\ell$ of $n$ can be ruled out once ${\rm ord}_{\ell}(q)$ is odd.

\begin{lem}\label{reduce}
Let $\ell$ be an odd prime divisor of $n$, so that  $n=\ell^sn_1$ with $\gcd(\ell, n_1)=1$.
Then
$|\Omega_n|=|\Omega_{n_1}|$
if and only if
${\rm ord}_{\ell}(q)$ is odd.
\end{lem}
\begin{proof}
Suppose first that ${\rm ord}_{\ell}(q)$ is odd. By Lemma~\ref{pairs},
we can  assume,  therefore,  that $C_0=\{0\}$, $C_{i_1}, C_{-i_1},\cdots, C_{i_{\rho}}, C_{-i_{\rho}} $,
are all the distinct $q$-cyclotomic cosets modulo $\ell^s$.
Taking a primitive $\ell^s$th root of unity $\eta$ in a suitable extension field of $\mathbb{F}_q$, we get
\begin{equation*}
X^{\ell^s}-1= (X-1)M_{i_1}(X)M_{-i_1}(X)
 \cdots M_{i_{\rho}}(X)M_{-i_{\rho}}(X),
\end{equation*}
with
$$
M_{i_h}(X)=\prod\limits_{j\in C_{i_h}}(X-\eta^j),~~~M_{-i_h}(X)=\prod\limits_{j\in C_{-i_h}}(X-\eta^j),
 \qquad h=1,\cdots,\rho,$$
all being monic irreducible in $\mathbb{F}_q[X]$.
It follows that
\begin{equation*}
X^{n}-1=\big(X^{n_1}\big)^{\ell^s}-1=(X^{n_1}-1)M_{i_1}(X^{n_1})M_{-i_1}(X^{n_1})
 \cdots M_{i_{\rho}}(X^{n_1})M_{-i_{\rho}}(X^{n_1}).
\end{equation*}
Clearly,
$$
M_{i_h}(X^{n_1})^*=M_{-i_h}(X^{n_1}),~~~\hbox{ for all $1\leq h\leq\rho$}.
$$
This implies that the polynomials $M_{i_1}(X^{n_1}),
 \cdots,  M_{-i_{\rho}}(X^{n_1})$ contribute nothing to the value of $|\Omega_n|$.
We are done for this direction.

Conversely, assume that $|\Omega_n|=|\Omega_{n_1}|$. Observe that
$X^{n_1}-1$ and $\frac{X^{\ell^s}-1}{X-1}$ are both divisors of $X^n-1$,
and that $\gcd(X^{n_1}-1, \frac{X^{\ell^s}-1}{X-1})=1$.
This actually means
that $\frac{X^{\ell^s}-1}{X-1}$ contributes nothing to the
value of $|\Omega_n|$.
We get the desired result from Lemma~\ref{pairs} directly.
\end{proof}

For the value $|\Omega_n|$, now we can assume that $n=2^mn'$ with
$n'=p_1^{r_1}p_2^{r_2}\cdots p_{k}^{r_k}$, where
$p_j$ are distinct odd primes and $r_j$
are positive integers such that ${\rm ord}_{p_j}(q)$ are even for
all $1\leq j\leq k$.
The following lemma characterizes  the relationship between the $q$-cyclotomic cosets modulo $n'$ and
the $q^2$-cyclotomic cosets modulo $n'$.
\begin{lem}\label{orbits}
Let $n'=p_1^{r_1}p_2^{r_2}\cdots p_{k}^{r_k}$, where
$p_j$ are distinct odd primes and $r_j$
are positive integers such that ${\rm ord}_{p_j}(q)$ are even for
all $1\leq j\leq k$.
Suppose that there are exactly $s$ distinct $q$-cyclotomic
cosets modulo $n'$. Then the number of
distinct $q^2$-cyclotomic
cosets modulo $n'$ is precisely given by $2s-1$.
\end{lem}
\begin{proof}
Let
$$
C_{z}=\{z\cdot q^j~(\bmod~n')\,|\,j=0,1,\cdots\}
$$
be any nonzero  $q$-cyclotomic coset modulo $n'$. It is clear that
$$
D_{z}=\{z\cdot q^{2j}~(\bmod~n')\,|\,j=0,1,\cdots\}~~~\hbox{and}~~~D_{zq}=\{zq\cdot q^{2j}~(\bmod~n')\,|\,j=0,1,\cdots\}
$$
are  $q^2$-cyclotomic cosets modulo $n'$. Obviously, $C_z=D_{z}\bigcup D_{zq}$.
To complete the proof,  it suffices to show that
$D_z\neq D_{zq}$.
Suppose otherwise that an integer $j$ can be found such that
$z\equiv zq\cdot q^{2j}~(\bmod~n')$.
This implies that  $\xi^{z}$ is an element in $\mathbb{F}_{q^{2j+1}}$,
where $\xi$ is a primitive $n'$th root of unity.
 This is impossible:
Without loss of generality, we can assume that $p_1$ is a prime divisor of ${\rm ord}(\xi^z)$,
i.e., $\mathbb{F}_{q^{2j+1}}$ contains a primitive $p_1$th root of unity.
On the other hand, by our assumption, $f_1={\rm ord}_{p_1}(q)$  is even, which gives $f_1\mid(2j+1)$, a contradiction.
\end{proof}

Let
$$
C_{i_h}=\{i_h\cdot q^j~(\bmod~n')\,|\,j=0,1,\cdots\}, ~~~\hbox{$1\leq h\leq \rho,$}
$$
be all the distinct nonzero $q$-cyclotomic cosets modulo $n'$.
From Lemma~\ref{orbits}, we know that each $q$-cyclotomic coset
$C_{i_s}$ is a disjoint union of two $q^2$-cyclotomic cosets:
$$
D_{i_h}=\{i_h\cdot q^{2j}~(\bmod~n')\,|\,j=0,1,\cdots\}~~~\hbox{and}~~~
D_{i_hq}=\{i_hq\cdot q^{2j}~(\bmod~n')\,|\,j=0,1,\cdots\}.
$$
At this point, we can give the irreducible factorization of $X^{n'}-1$ over $\mathbb{F}_{q^2}$,
where $\mathbb{F}_{q^2}$ is the extension field over $\mathbb{F}_q$  such that $[\mathbb{F}_{q^2}\,:\,\mathbb{F}_q]=2$.
Let $\eta$ be a primitive $n'$th root of unity in some extension field
of $\mathbb{F}_{q^2}$. Then
\begin{equation*}\label{simple-irr-decomposition2}
X^{n'}-1=(X-1) N_{i_1}(X)N_{i_1q}(X)N_{i_2}(X)N_{i_2q}(X)
 \cdots N_{i_\rho}(X)N_{i_\rho q}(X),
\end{equation*}
with
$$
N_{i_h}(X)=\prod\limits_{u\in D_{i_h}}(X-\eta^u),~~N_{i_hq}(X)=\prod\limits_{u\in D_{i_hq}}(X-\eta^u), ~~1\leq h\leq \rho,
$$
all being monic irreducible in $\mathbb{F}_{q^2}[X]$.
Note that
$
X^{n'}-1=(X-1)\prod\limits_{h=1}^\rho M_{i_h}(X)
$
gives the monic irreducible factorization of $X^{n'}-1$ over $\mathbb{F}_q$, where
$M_{i_h}(X)=N_{i_h}(X)N_{i_h q}(X)$, $1\leq h\leq \rho$.

Before giving our results, we  make the following observation.
Assume that
$$
X^{n'}-1=(X-1)M_{i_1}(X)\cdots M_{i_u}(X)M_{j_1}(X)M_{-j_1}(X)\cdots M_{j_v}(X)M_{-j_v}(X),
$$
where $M_{i_k}(X)$ are self-reciprocal monic irreducible factors for $1\leq k\leq u$, while $M_{j_s}(X)$
and $M_{-j_s}(X)$ are reciprocal polynomial pairs for $1\leq s\leq v$.
We can further assume that
$X-1, N_{i_1}, N_{i_1q}, \cdots, N_{i_b}, N_{i_bq},$ are self-reciprocal monic irreducible factors
of $X^{n'}-1$ over $\mathbb{F}_{q^2}$. That is to say,  $2b+1$ is the number of  all self-reciprocal monic irreducible factors of
$X^{n'}-1$ over $\mathbb{F}_{q^2}$, i.e., $|\overline \Omega_{n'}|=2b+1$.
Now the irreducible factorization of $X^{n'}-1$ over $\mathbb{F}_{q^2}$ can be given as follows:
\begin{equation}\label{important}
X^{n'}-1=(X-1)N_{i_1}N_{i_1q}\cdots N_{i_b}N_{i_bq}
N_{i_{(b+1)}}N_{i_{(b+1)}q}\cdots N_{i_{u}}N_{i_{u}q}
N_{j_1}N_{j_1q}\cdots  N_{j_v}N_{j_vq}N_{-j_v}N_{-j_vq}.
\end{equation}
We assert that $N_{-i_j}=N_{i_jq}$ for $(b+1)\leq j\leq u$;
this is because for $(b+1)\leq j\leq u$, $N_{i_j}N_{i_jq}=M_{i_j}=M_{i_j}^*=N_{-i_j}N_{-i_jq}$,
but $N_{i_j}\neq N_{-i_j}$ by assumption.

Let $f(X)$ be a polynomial in $\mathbb{F}_q[X]$  with leading coefficient
$a_n\neq0$. Recall from Notation 3.1 that $\hat f(X)=a_n^{-1}f(X)$ is a monic polynomial over $\mathbb{F}_q$.

\begin{Theorem}\label{reduce2}
With respect to the above notation, we then have
 \begin{equation*}\label{omega}
|\Omega_{n}|=|\Omega_{2^mn'}|=\left\{
                   \begin{array}{ll}
2|\Omega_{n'}|, & \hbox{if $m=1$,  or $m\geq2$ and $4\mid(q-1)$;}\\
2|\Omega_{n'}|+(2^{\min\{m,a\}-1}-1)(2|\Omega_{n'}|-|\overline \Omega_{n'}|), & \hbox{if  $m\geq2$ and $4\nmid(q-1)$.}
                   \end{array}
                 \right.
\end{equation*}
Here, for the case $4\nmid(q-1)$, $a$ is the unique integer such that $2^a\Vert(q+1)$.
\end{Theorem}
\begin{proof}
If $m=1$,
the result follows trivially. Indeed, from $X^n-1=(X^{n'}-1)(X^{n'}+1)$,
we easily get $|\Omega_{n}|=2|\Omega_{n'}|$.
We  prove  by  induction  on $m$ for the case $m\geq2$ and $q\equiv1~(\bmod~4)$.
If $m=2$,
then $X^{4n'}-1=(X^{n'}-1)(X^{n'}+1)(X^{n'}-\alpha)(X^{n'}+\alpha)$,
where $\alpha$ is a primitive fourth root of unity in $\mathbb{F}_q$.
Observe that  $(X^{n'}-\alpha)^*=X^{n'}-\alpha^{-1}=X^{n'}+\alpha$,
which implies that  $X^{n'}-\alpha$ and $X^{n'}+\alpha$ contribute nothing to the
value of $|\Omega_n|$. Hence the required result follows  directly.
For the inductive step, we write
$$
X^n-1=(X^{2^{m-1}n'}-1)(X^{2^{m-2}n'}-\alpha)(X^{2^{m-2}n'}+\alpha).
$$
Similar  reasoning  then  shows  that $X^{2^{m-2}n'}-\alpha$ and $X^{2^{m-2}n'}+\alpha$
contribute nothing to the
value of $|\Omega_n|$.
Thus $|\Omega_{n}|=|\Omega_{2^{m-1}n'}|=2|\Omega_{n'}|$ by induction.

We are left with the case $m\geq2$ and $4\nmid (q-1)$.
We use Lemma~\ref{irr-trinomial} to prove this result.  Assume  first that $2\leq m\leq a$. From (\ref{first}),
$$
X^{2^m}-1=(X-1)(X+1)\prod\limits_{i=1}^{m-1}\prod\limits_{h\in H_{i}}(X^2-2hX+1).
$$
Then
\begin{equation*}
X^{2^mn'}-1=(X^{n'}-1)(X^{n'}+1)\prod\limits_{i=1}^{m-1}\prod\limits_{h\in H_{i}}(X^{2n'}-2hX^{n'}+1).
\end{equation*}
The irreducible factorization of $X^{2n'}-2hX^{n'}+1$ over $\mathbb{F}_{q^2}$,  $1\leq i\leq m-1$ and $h\in H_i$, can be described via the
$q^2$-cyclotomic cosets modulo $n'$, as we will show shortly.
Since $X^{2}-2hX+1$ is an irreducible factor of $X^{2^m}-1$ over $\mathbb{F}_q$,
there exists an element $\beta_h$ in the Sylow 2-subgroup of $\mathbb{F}_{q^2}^*$ such that
$X^{2}-2hX+1=(X-\beta_h)(X-\beta_h^{-1})$.
Note that $\beta_h^q=\beta_h^{-1}$.
We then have
$$
X^{2n'}-2hX^{n'}+1=(X^{n'}-\beta_h)(X^{n'}-\beta_h^{-1}).
$$
On the one hand,  for each element $\beta_h$,
there exists a unique element $\lambda_h$ in the Sylow $2$-subgroup of $\mathbb{F}_{q^2}^*$ such that
$\lambda_h^{n'}\beta_h=1$. We  also note that $\lambda_h^q=\lambda_h^{-1}$ for each $h\in H_i$,
because $\lambda_h^{qn'}=\beta_h^{-q}=\beta_h=\lambda_h^{-n'}$.
On the other hand, (\ref{important}) gives the monic irreducible factorization of $X^{n'}-1$ over $\mathbb{F}_{q^2}$.
We substitute $X$ for $\lambda_hX$ in (\ref{important}) to obtain the  monic irreducible factorization of
$X^{n'}-\beta_h$  over $\mathbb{F}_{q^2}$:
$$
X^{n'}-\beta_h= (X-\lambda_h^{-1})\hat N_{i_1}(\lambda_hX)
 \cdots\hat N_{i_{b}q}(\lambda_hX)\hat N_{i_{(b+1)}}(\lambda_hX)\cdots \hat N_{i_{u}q}(\lambda_hX)
\hat N_{j_1}(\lambda_hX)\cdots \hat N_{-j_{v}q}(\lambda_hX).
$$
Similarly,
$$
X^{n'}-\beta_h^{-1}= (X-\lambda_h)\hat N_{i_1}(\lambda_h^{-1}X)
 \cdots\hat N_{i_{b}q}(\lambda_h^{-1}X)\hat N_{i_{(b+1)}}(\lambda_h^{-1}X)\cdots \hat N_{i_{u}q}(\lambda_h^{-1}X)
\hat N_{j_1}(\lambda_h^{-1}X)\cdots \hat N_{-j_{v}q}(\lambda_h^{-1}X).
$$
Now it is easy to check that the  polynomial $(X-\lambda_h^{-1})(X-\lambda_h)$
is self-reciprocal monic irreducible over $\mathbb{F}_q$.
For $1\leq k\leq b$, we  assert that
$\hat N_{i_k}(\lambda_hX)\hat N_{i_k q}(\lambda_h^{-1}X)$
and
$\hat N_{i_k q}(\lambda_hX)\hat N_{i_k}(\lambda_h^{-1}X)$ are irreducible over $\mathbb{F}_q$
 and form a reciprocal polynomial pair:
\begin{equation}\label{polynomial-pairs}
\Big(\hat N_{i_k}(\lambda_hX)\hat N_{{i_k} q}(\lambda_h^{-1}X)\Big)^*=
\hat N_{{i_k}q}(\lambda_hX)\hat N_{i_k}(\lambda_h^{-1}X).
\end{equation}
Assuming  that $\eta$ is a primitive
$n'$th root of unity in some extension field of $\mathbb{F}_{q^2}$, then
$$
\hat N_{i_k}(\lambda_hX)\hat N_{i_k q}(\lambda_h^{-1}X)=\prod\limits_{j\in D_{i_k}}(X-\lambda_h^{-1}\eta^j)\cdot
\prod\limits_{j\in D_{i_kq}}(X-\lambda_h\eta^j).
$$
For every $j\in D_{i_k}$, $\lambda_h^{-q}\eta^{jq}=\lambda_h\eta^{jq}$ is a root of $\hat N_{i_k q}(\lambda_h^{-1}X)$;
for every $j\in D_{i_kq}$, $\lambda_h^{q}\eta^{jq}=\lambda_h^{-1}\eta^{jq}$ is a root of $\hat N_{i_k}(\lambda_hX)$.
In particular, the roots of $\hat N_{i_k}(\lambda_hX)\hat N_{i_k q}(\lambda_h^{-1}X)$ are invariant under the
action of the Galois group ${\rm Gal}(\mathbb{F}_{q^2}/\mathbb{F}_q)$.
It follows that $\hat N_{i_k}(\lambda_hX)\hat N_{i_k q}(\lambda_h^{-1}X)$ is a polynomial over $\mathbb{F}_q$.
Moreover, $\hat N_{i_k}(\lambda_hX)\hat N_{i_k q}(\lambda_h^{-1}X)$ is irreducible over $\mathbb{F}_q$,
since  $\hat N_{i_k}(\lambda_hX)\not\in\mathbb{F}_q[X]$ and $\hat N_{i_k q}(\lambda_h^{-1}X)\not\in \mathbb{F}_q[X]$.
Similar reasoning shows that $\hat N_{i_k q}(\lambda_hX)\hat N_{i_k}(\lambda_h^{-1}X)$ is irreducible over $\mathbb{F}_q$.
We are left with proving Formula~(\ref{polynomial-pairs}).
Note that
$$
\hat N_{i_k}(\lambda_hX)^*=\Big(\prod\limits_{j\in D_{i_k}}(X-\lambda_h^{-1}\eta^j)\Big)^*=\prod\limits_{j\in D_{-i_k}}(X-\lambda_h\eta^j)
=\prod\limits_{j\in D_{i_k}}(X-\lambda_h\eta^j)=\hat N_{i_k}(\lambda_h^{-1}X).
$$
The third equality holds because $D_{-i_k}=D_{i_k}$ for $1\leq k\leq b$. Similarly, $\hat N_{i_kq}(\lambda_h^{-1}X)^*=\hat N_{i_kq}(\lambda_hX)^*$.
Thus, Formula~(\ref{polynomial-pairs})  has been established.

Using similar arguments,
for $(b+1)\leq k\leq u$, $\hat N_{i_k}(\lambda_hX)\hat N_{{i_k} q}(\lambda_h^{-1}X)$ and
$\hat N_{{i_k} q}(\lambda_hX)\hat N_{i_k }(\lambda_h^{-1}X)$ are self-reciprocal monic irreducible polynomials over $\mathbb{F}_q$:
$$
\Big(\hat N_{i_k}(\lambda_hX)\hat N_{{i_k} q}(\lambda_h^{-1}X)\Big)^*=
\hat N_{i_k}(\lambda_hX)\hat N_{{i_k} q}(\lambda_h^{-1}X),
$$
$$
\Big(\hat N_{{i_k} q}(\lambda_hX)\hat N_{i_k }(\lambda_h^{-1}X)\Big)^*=
\hat N_{{i_k} q}(\lambda_hX)\hat N_{i_k }(\lambda_h^{-1}X).
$$
Finally for $1\leq k\leq v$,
$$
\Big(\hat N_{j_k}(\lambda_hX)\hat N_{{j_k} q}(\lambda_h^{-1}X)\Big)^*=
\hat N_{-j_k}(\lambda_h^{-1}X)\hat N_{-{j_k} q}(\lambda_hX),
$$
$$
\Big(\hat N_{j_{k}q}(\lambda_hX)\hat N_{j_k}(\lambda_h^{-1}X)\Big)^*=
\hat N_{-j_{k}q}(\lambda_h^{-1}X)\hat N_{-j_k}(\lambda_hX).
$$
It is readily seen that, for every $1\leq i\leq m-1$ and $h\in H_i$,
there are exactly  $1+2(u-b)$ self-reciprocal monic irreducible factors of $X^{2n'}-2hX^{n'}+1$
over $\mathbb{F}_q$.
Consequently,
$
\prod_{i=1}^{m-1}\prod_{h\in H_{i}}(X^{2n'}-2hX^{n'}+1)
$
contributes $(2^{m-1}-1)(2|\Omega_{n'}|-|\overline \Omega_{n'}|)$
self-reciprocal monic irreducible factors to the value of $|\Omega_{n}|$,
since $|\Omega_{n'}|=1+u$ and $|\overline\Omega_{n'}|=1+2b$.

We are left to consider the case $m\geq a+1$.
It follows from (\ref{second}) that
$$
X^{2^mn'}-1=(X^{n'}-1)(X^{n'}+1)\prod\limits_{h\in H_{i}, ~\atop{1\leq i\leq(a-1)}}(X^{2n'}-2hX^{n'}+1)
\prod\limits_{h\in H_{a},~\atop{0\leq k\leq(m-a-1)}}(X^{2^{k+1}n'}-2hX^{2^{k}n'}-1).
$$
We just note that the last term
$\prod\limits_{h\in H_{a},~\atop{0\leq k\leq(m-a-1)}}(X^{2^{k+1}n'}-2hX^{2^{k}n'}-1)$,  contributes nothing to the
value of $|\Omega_n|$.
Similar  reasoning yields our desired result.
\end{proof}

As an immediate application of Theorem~\ref{reduce2}, a general formula for the value of $|\Omega_{2^m\ell^s}|$
can be given explicitly, as we show below.
\begin{Corollary}
Let $\ell$ be an odd prime integer coprime with $q$, and let $s$ be a positive integer.
If ${\rm ord}_\ell(q)$ is odd, then
$|\Omega_{2^m\ell^s}|=|\Omega_{2^m}|$, where $|\Omega_{2^m}|$ was explicitly given by (\ref{omega1}).
Otherwise, we  have:

{\rm (i)}~If $m=1$, or $m\geq2$ and $4\mid(q-1)$, then
$|\Omega_{2^m\ell^s}|=2|\Omega_{\ell^s}|=2\sum\limits_{d=0}^s\frac{\phi(\ell^d)}{{\rm ord}_{\ell^d}(q)}$.

{\rm (ii)}~If  $m\geq2$ and $4\nmid(q-1)$, then $|\Omega_{2^m\ell^s}|$ is equal to
$$
2\sum\limits_{d=0}^s\frac{\phi(\ell^d)}{{\rm ord}_{\ell^d}(q)}+(2^{\min\{m,a\}-1}-1)(2\sum\limits_{d=0}^s\frac{\phi(\ell^d)}{{\rm ord}_{\ell^d}(q)}-|\overline \Omega_{\ell^s}|).
$$
Here, if $2\,||\, {\rm ord}_\ell(q)$, then  $|\overline \Omega_{\ell^s}|=1$; otherwise,
$|\overline \Omega_{\ell^s}|=2\sum\limits_{d=0}^s\frac{\phi(\ell^d)}{{\rm ord}_{\ell^d}(q)}-1$.
\end{Corollary}

\subsection{An enumeration formula for $|\Omega_{\ell_1^{r_1}\ell_2^{r_2}}|$ }

In  this subsection,  we give a general formula for the value of $|\Omega_{\ell_1^{r_1}\ell_2^{r_2}}|$, where
$\ell_1, \ell_2$ are distinct odd primes coprime with $q$,  and $r_1, r_2$ are positive integers.
We set ${\rm ord}_{\ell_i^{r_i}}(q)=2^{a_i}f_i$ with $\gcd(2, f_i)=1$,
$i=1,2$. By Lemma~\ref{reduce}, we can assume that $a_1\geq1$ and $a_2\geq1$.

If $a_1=a_2\geq1$, we claim that
all monic irreducible factors of $X^{\ell_1^{r_1}\ell_2^{r_2}}-1$ over $\mathbb{F}_q$ are self-reciprocal,
and hence $|\Omega_{\ell_1^{r_1}\ell_2^{r_2}}|=
\sum\limits_{d\,\mid\,{\ell_1^{r_1}\ell_2^{r_2}}}\frac{\phi(d)}{{\rm ord}_{d}(q)}$.
To this end, it suffices to prove that there exists some integer $i_0$ such that $q^{i_0}\equiv-1~(\bmod~\ell_1^{r_1}\ell_2^{r_2})$.
Since $q^{2^{a_1-1}f_1}\equiv-1~(\bmod~\ell_1^{r_1})$ and $q^{2^{a_1-1}f_2}\equiv-1~(\bmod~\ell_2^{r_2})$.
it follows that $q^{2^{a_1-1}f_1f_2}\equiv-1~(\bmod~\ell_1^{r_1})$ and $q^{2^{a_1-1}f_1f_2}\equiv-1~(\bmod~\ell_2^{r_2})$.
We then have $q^{2^{a_1-1}f_1f_2}\equiv-1~(\bmod~\ell_1^{r_1}\ell_2^{r_2})$, as claimed.

Thus, without loss of generality, we are left to consider the case $1\leq a_1<a_2$.

To compute $|\Omega_{\ell_1^{r_1}\ell_2^{r_2}}|$, we need to know the relationship between
$q$-cyclotomic cosets modulo $\ell_2^{r_2}$ and
$q^{2^{a_1}f_1}$-cyclotomic cosets modulo $\ell_2^{r_2}$.
In fact, for any nonzero $q$-cyclotomic coset modulo $\ell_2^{r_2}$,
$$
C_{j_k}=\{j_k\cdot q^j\pmod{\ell_2^{r_2}}\,|\,j=0,1,\cdots\},
$$
we assert that $C_{j_k}$ is a disjoint union of $q^{2^{a_1}f_1}$-cyclotomic cosets modulo $\ell_2^{r_2}$:
$$
C_{j_k}=D_{j_k}\bigcup D_{j_kq}\bigcup\cdots\bigcup D_{j_kq^{d_k-1}}, ~~
$$
where $D_{j_k}, D_{j_kq}, \cdots, D_{j_kq^{d_k-1}}$ are
$q^{2^{a_1}f_1}$-cyclotomic cosets modulo $\ell_2^{r_2}$ and
$d_k$ is the smallest positive integer such that $D_{j_kq^{d_k}}=D_{j_k}$.
This can be seen as follows. We can always divide $C_{j_k}$ into unions of
$q^{2^{a_1}f_1}$-cyclotomic cosets (not necessary disjoint):
$$
C_{j_k}=D_{j_k}\bigcup D_{j_kq}\bigcup\cdots\bigcup D_{j_kq^{d_k-1}}\bigcup\cdots\bigcup D_{j_kq^{{2^{a_1}f_1}-1}}.
$$
Note that $D_{j_kq^{{2^{a_1}f_1}}}=D_{j_k}$ as $q^{2^{a_1}f_1}$-cyclotomic cosets
modulo $\ell_2^{r_2}$. Now assume that
$d_k$ is the smallest positive integer such that $D_{j_kq^{d_k}}=D_{j_k}$.
It is clear that every term between $D_{j_kq^{d_k}}$ and
$D_{j_kq^{{2^{a_1}f_1}-1}}$ is exactly equal to one term of
$D_{j_k}, D_{j_kq}, \cdots, D_{j_kq^{d_k-1}}$. Thus, we get the desired decomposition.

In the following,  we first give the irreducible factorization of $X^{\ell_1^{r_1}\ell_2^{r_2}}-1$ over $\mathbb{F}_{q^{2^{a_1}f_1}}$.
Then we recombine the irreducible factors such that each of them is actually  irreducible  over $\mathbb{F}_{q}$.
The following  well-known fact from Galois theory will be used  (e.g., see \cite[Theorem 4.14]{Jacobson}):
Let $\mathbb{E}$ be a finite extension field over $\mathbb{F}_q$. Let $\alpha_1,\alpha_2,\cdots, \alpha_k$ be distinct
elements of $\mathbb{E}$ such that $\{\alpha_1,\alpha_2,\cdots, \alpha_k\}=\{\alpha_1^q,\alpha_2^q,\cdots, \alpha_k^q\}$.
Then $f(X)=(X-\alpha_1)(X-\alpha_2)\cdots(X-\alpha_k)$ is a monic polynomial over $\mathbb{F}_q$; if, in addition,  for any two elements
$\alpha_i, \alpha_j$, there exists an integer $s$ such that $\alpha_i^{q^s}=\alpha_j$,
then $f(X)$ is a monic irreducible polynomial over $\mathbb{F}_q$.

Assume that
$
C_{i_h}=\{i_h\cdot q^j\pmod{\ell_1^{r_1}}\,|\,j=0,1,\cdots\}, ~~\hbox{$1\leq h\leq \rho,$}
$
are all the distinct nonzero $q$-cyclotomic cosets modulo $\ell_1^{r_1}$.
Let
$$
X^{\ell_1^{r_1}}-1= (X-1)M_{i_1}(X)
\cdots M_{i_{\rho}}(X)
$$
be the monic irreducible factorization  of $X^{\ell_1^{r_1}}-1$ over  $\mathbb{F}_q$,  where
$M_{i_h}(X)=\prod\limits_{j\in C_{i_h}}(X-\eta^j)$ and $\eta$
is a primitive $\ell_1^{r_1}$th root of unity in $\mathbb{F}_{q^{2^{a_1}f_1}}$.
It follows that
$$
X^{\ell_1^{r_1}\ell_2^{r_2}}-1= (X^{\ell_2^{r_2}}-1)M_{i_1}(X^{\ell_2^{r_2}})
\cdots M_{i_{\rho}}(X^{\ell_2^{r_2}}).
$$
Note that $M_{i_h}(X^{\ell_2^{r_2}})^*=M_{i_h}(X^{\ell_2^{r_2}})$ for all $1\leq h\leq\rho$, because $M_{i_h}(X)^*=M_{i_h}(X)$
by Remark~\ref{even}.
We need to answer this question:
how many self-reciprocal monic irreducible factors of each $M_{i_h}(X^{\ell_2^{r_2}})$
contribute to $|\Omega_{\ell_1^{r_1}\ell_2^{r_2}}|$ ?
The answer is precisely equal to 1, as we will show shortly.

Now,  assuming that $\deg M_{i_1}(X)=t_1$,  in $\mathbb{F}_{q^{2^{a_1}f_1}}[X]$,
$$
M_{i_1}(X^{\ell_2^{r_2}})=\prod\limits_{s\in C_{i_1}}(X^{\ell_2^{r_2}}-\eta^s)
=(X^{\ell_2^{r_2}}-\eta^{i_1})(X^{\ell_2^{r_2}}-\eta^{i_1q})\cdots(X^{\ell_2^{r_2}}-\eta^{i_1q^{t_1-1}}).
$$
Let $C_{j_k}$, $1\leq k\leq\nu$,  be all the distinct nonzero $q$-cyclotomic cosets modulo $\ell_2^{r_2}$, and
let $N_{j_k}(X)$ be the  monic irreducible factor of $X^{\ell_2^{r_2}}-1$ over $\mathbb{F}_q$ corresponding to $C_{j_k}$.
It follows that $N_{j_k}(X)$ splits into $d_k$ irreducible factors over $\mathbb{F}_{q^{2^{a_1}f_1}}$,
$N_{j_k}(X)=D_{j_k}(X)D_{j_kq}(X)\cdots D_{j_kq^{d_k-1}}(X)$, where $D_{j_k}(X),\cdots, D_{j_kq^{d_k-1}}(X)$ are monic
irreducible factors over $\mathbb{F}_{q^{2^{a_1}f_1}}$ corresponding to the $q^{2^{a_1}f_1}$-cyclotomic cosets $D_{j_k},\cdots, D_{j_kq^{d_k-1}},$
respectively.
We then have the monic irreducible factorization of $X^{\ell_2^{r_2}}-1$ over $\mathbb{F}_{q^{2^{a_1}f_1}}$ as follows:
\begin{equation}\label{factorizationl}
X^{\ell_2^{r_2}}-1=(X-1)\prod\limits_{k=1}^{\nu}D_{j_k}(X)D_{j_kq}(X)\cdots D_{j_kq^{d_k-1}}(X).
\end{equation}
We see that for each $s\in C_{i_1}$,  there exists a unique element $\zeta_s$ in the Sylow $\ell_1$-subgroup of
$\mathbb{F}_{q^{2^{a_1}f_1}}^*$ such that $\zeta_s^{\ell_2^{r_2}}\eta^s=1$.
Indeed, this is because $\eta$ is an element of the Sylow $\ell_1$-subgroup of
$\mathbb{F}_{q^{2^{a_1}f_1}}^*$ and $\gcd(\ell_1,\ell_2)=1$.
Thus, we substitute $X$ for $\zeta_sX$ in (\ref{factorizationl}) to obtain
$$
X^{\ell_2^{r_2}}-\eta^s=(X-\zeta_s^{-1})\prod\limits_{k=1}^{\nu}\hat D_{j_k}(\zeta_sX)\hat D_{j_kq}(\zeta_sX)\cdots \hat D_{j_kq^{d_k-1}}(\zeta_sX),
$$
the monic irreducible factorization of $X^{\ell_2^{r_2}}-\eta^s$ over $\mathbb{F}_{q^{2^{a_1}f_1}}$.
At this point, the monic irreducible factorization of $M_{i_1}(X^{\ell_2^{r_2}})$ over $\mathbb{F}_{q^{2^{a_1}f_1}}$ is given by
\begin{equation*}
\begin{split}
M_{i_1}(X^{\ell_2^{r_2}})
&=\prod\limits_{s\in C_{i_1}}\Big(
(X-\zeta_s^{-1})\prod\limits_{k=1}^{\nu}\hat D_{j_k}(\zeta_sX)\hat D_{j_kq}(\zeta_sX)\cdots \hat D_{j_kq^{d_k-1}}(\zeta_sX)\Big)\\
&=\prod\limits_{s\in C_{i_1}}(X-\zeta_s^{-1})\cdot
\prod\limits_{s\in C_{i_1}}\Big(\prod\limits_{k=1}^{\nu}\hat D_{j_k}(\zeta_sX)\hat D_{j_kq}(\zeta_sX)\cdots \hat D_{j_kq^{d_k-1}}(\zeta_sX)\Big).
\end{split}
\end{equation*}
For any integer $k$, since $\zeta_{i_1}^{\ell_2^{r_2}}\eta^{i_1}=1$ and  $\zeta_{i_1q^k}^{\ell_2^{r_2}}\eta^{i_1q^k}=1$, so
$(\zeta_{i_1}^{q^k})^{\ell_2^{r_2}}=\eta^{-q^ki_1}=\zeta_{i_1q^k}^{\ell_2^{r_2}}$, which implies
$\zeta_{i_1}^{q^k}=\zeta_{i_1 q^k}$.
It follows that
$$
\prod\limits_{s\in C_{i_1}}(X-\zeta_s^{-1})=(X-\zeta_{i_1}^{-1})(X-\zeta_{i_1q}^{-1})\cdots(X-\zeta_{i_1q^{t_1-1}}^{-1})
$$
is a  monic  irreducible polynomial over $\mathbb{F}_q$. Moreover, $\prod\limits_{s\in C_{i_1}}(X-\zeta_s^{-1})$ is self-reciprocal,
since it is a divisor of $X^{\ell_1^{r_1}}-1$ over $\mathbb{F}_q$.
Now for any positive integer $k$, $1\leq k\leq \nu$, we analyze the polynomial
$
F_k(X)=\prod\limits_{s\in C_{i_1}}\big(\hat D_{j_k}(\zeta_sX)\hat D_{j_kq}(\zeta_sX)\cdots \hat D_{j_kq^{d_k-1}}(\zeta_sX)\big).
$
The polynomial $F_k(X)$ can be rewritten as follows:
$$
F_k(X)
=\prod\limits_{r=0}^{t_1-1}\big(\hat D_{j_k}(\zeta_{i_1q^{t_1-r}}X)\hat D_{j_kq}(\zeta_{i_1q^{t_1-r+1}}X)\cdots \hat D_{j_kq^{d_k-1}}(\zeta_{i_1q^{t_1-r+t_1-1}}X)\big).
$$
Here,  $\zeta_{i_1q^{t_1}}=\zeta_{i_1}$,
namely the exponents of $q$ are calculated modulo $t_1$.
Recall that $D_{j_kq^{d_k}}=D_{j_k}$. We deduce that, for any $0\leq r\leq t_1-1$,
$$
G_r(X)=\hat D_{j_k}(\zeta_{i_1q^{t_1-r}}X)\hat D_{j_kq}(\zeta_{i_1q^{t_1-r+1}}X)\cdots \hat D_{j_kq^{d_k-1}}(\zeta_{i_1q^{t_1-r+t_1-1}}X)
$$
is irreducible over $\mathbb{F}_q$.
Since by assumption, $1\leq a_1<a_2$, it follows  that  ${\rm ord}_{\ell_2^{r_2}}(q^{2^{a_1}f_1})$ is even.  By Remark~\ref{even},
$D_{j_kq^i}(X)^*=D_{j_kq^i}(X)$ for every $0\leq i\leq d_k-1$.  We then have
$$
G_r(X)^*
=\hat D_{j_k}(\zeta_{i_1q^{t_1-r}}^{-1}X)\hat D_{j_kq}(\zeta_{i_1q^{t_1-r+1}}^{-1}X)\cdots \hat D_{j_kq^{d_k-1}}(\zeta_{i_1q^{t_1-r+t_1-1}}^{-1}X).
$$
We claim  that $G_r(X)^*\neq G_r(X)$ for all $0\leq r\leq t_1-1$. Because
$\eta^{i_1q^{t_1-r}}$ is neither 1 nor -1, we see that $\zeta_{i_1q^{t_1-r}}$ is neither 1 nor -1.
This implies that $\hat D_{j_k}(\zeta_{i_1q^{t_1-r}}X)\neq\hat D_{j_k}(\zeta_{i_1q^{t_1-r}}^{-1}X)$.
Suppose otherwise that $\hat D_{j_k}(\zeta_{i_1q^{t_1-r}}X)=\hat D_{j_kq}(\zeta_{i_1q^{t_1-r+1}}^{-1}X)$.
This leads to $\zeta_{i_1q^{t_1-r}}=\zeta_{i_1q^{t_1-r+1}}^{-1}$, and hence
$\hat D_{j_k}(\zeta_{i_1q^{t_1-r}}X)=\hat D_{j_kq}(\zeta_{i_1q^{t_1-r}}X)$.
This is a contradiction,  since we would obtain
$\hat D_{j_k}(X)=\hat D_{j_kq}(X)$.
Each of the remaining  factors of $G_r(X)^*$ does likewise in turn, proving the claim.

Summarizing the discussions above, we have the following.
\begin{Theorem}
Let $\ell_1, \ell_2$ be distinct odd primes coprime with $q$, and
let $r_1, r_2$ be positive integers. Put ${\rm ord}_{\ell_i^{r_i}}(q)=2^{a_i}f_i$ with $\gcd(2, f_i)=1$,
$i=1,2$. We then have:

{\rm (i)}~ If $a_1=0$ (resp. $a_2=0$), then $|\Omega_{\ell_1^{r_1}\ell_2^{r_2}}|=|\Omega_{\ell_2^{r_2}}|$
{\rm(}resp. $|\Omega_{\ell_1^{r_1}\ell_2^{r_2}}|=|\Omega_{\ell_1^{r_1}}|${\rm)}.

Otherwise, there are two subcases:

{\rm (ii.a)}~If $a_1=a_2\geq1$, then $|\Omega_{\ell_1^{r_1}\ell_2^{r_2}}|=
\sum\limits_{d\,\mid\,{\ell_1^{r_1}\ell_2^{r_2}}}\frac{\phi(d)}{{\rm ord}_{d}(q)}$.

{\rm (ii.b)}~If $a_1\neq a_2$,   $a_1\geq1$ and $a_2\geq1$, then $|\Omega_{\ell_1^{r_1}\ell_2^{r_2}}|=
\sum\limits_{d_1=0}^{r_1}\frac{\phi(\ell_1^{d_1})}{{\rm ord}_{\ell_1^{d_1}}(q)}
+\sum\limits_{d_2=0}^{r_2}\frac{\phi(\ell_2^{d_2})}{{\rm ord}_{\ell_2^{d_2}}(q)}-1$.
\end{Theorem}

\noindent{\bf Acknowledgements}
The first and  third authors  thank NSFC for the support from  Grant No.~11171370.
The research of the first and second authors is also partially supported by Nanyang
Technological University's research grant number M4080456.


\begin{thebibliography}{21}
\bibitem{Abualrub}
T. Abualrub, R. Oehmke, On the generators of $\mathbb{Z}_4$ cyclic codes of length $2^e$, IEEE Trans. Inform. Theory,  {\bf 9}(2003),
 2126-2133.




\bibitem{Blake}
I. F. Blake, S. Gao, R. C. Mullin, Explicit factorization of $X^{2^k}+1$ over $F_p$ with
prime $p\equiv3 ~(\bmod~4)$, Appl. Algebra Engrg. Comm. Comput., {\bf 4}(1993), 89-94.




\bibitem{Chen}
B. Chen,  L. Li, R. Tuerhong, Explicit factorization of $X^{2^mp^n}-1$ over a finite field, Finite Fields Appl., {\bf 24}(2013),
95-104.

\bibitem{Chen2}
B. Chen,  H. Liu, G. Zhang, A class of minimal cyclic codes over finite fields,  Designs, Codes and Crypt.,
(2013),
DOI: 10.1007/s10623-013-9857-9.

\bibitem{Dinh04}
H. Q. Dinh, S. R. L\'{o}pez-Permouth, Cyclic and negacyclic codes over finite chain rings. IEEE Trans. Inform.
Theory,   {\bf50}(2004), 1728-1744.




\bibitem{Dinh10}
H. Q. Dinh, Constacyclic codes of length $p^s$ over $F_{p^m}+uF_{p^m}$, J. Algebra,   {\bf324}(2010),  940-950.


\bibitem{Dougherty}
S. T. Dougherty, S. Ling, Cyclic codes over $\mathbb{Z}_4$ of even length, Designs, Codes and Crypt., {\bf 2}(2006),  127-153.




\bibitem{Hammons}
A. R. Hammons, Jr. P. V. Kumar, A. R. Calderbank, N. J. A. Sloane, P. Sol\'{e},  The $\mathbb{Z}_4$-
linearity of Kerdock, Preparata, Goethals and related codes, IEEE Trans. Inform. Theory,
{\bf 40}(1994), 301-319.




\bibitem{Huffman}
W. C. Huffman, V. Pless, Fundamentals of Error-Correcting Codes, Cambridge University Press, Cambridge, 2003.

\bibitem{Jacobson}
N. Jacobson, Basic Algebra {\bf I}, second edition,   Freeman, San Francisco,  1985.

\bibitem{Jia}
Y. Jia, S. Ling, C. Xing, On self-dual cyclic codes over finite fields, IEEE Trans. Inform. Theory,  {\bf 57}(2011),
2243-2251.



\bibitem{Kai1}
X. Kai, S. Zhu, On cyclic self-dual codes, Appl. Algebra Engrg. Comm. Comput.,  {\bf 19}(2008), 509-525.





\bibitem{Kai2}
X. Kai, S. Zhu, Negacyclic self-dual codes over finite chain rings, Designs, Codes and Crypt.,
{\bf 62}(2012), 161-174.


\bibitem{Kanwar}
P. Kanwar, S. R. L\'{o}pez-Permouth, Cyclic codes over the integers
modulo $p^m$,  Finite Fields Appl., {\bf3}(1997),  334-352.





\bibitem{li}
R. Lidl, H. Niederreiter, Finite Fields, Cambridge University Press, Cambridge, 2008.


\bibitem{MacWilliams}
F. J. MacWilliams, N. J. A. Sloane, The Theory of Error-Correcting Codes,  10th Impression, North-Holland, Amsterdam, 1998.

\bibitem{McDonald}
B.  R. McDonald,  Finite Rings with Identity,   Marcel Dekker Press, New York,  1974.


\bibitem{Norton}
G. Norton,  A. S\u{a}l\u{a}gean-Mandache, On the structure of linear cyclic
codes over finite chain rings,  Appl. Algebra Engrg. Comm. Comput.,
{\bf 10}(2000), 489-506.

\bibitem{Pless}
V. Pless, Z. Qian, Cyclic codes and quadratic residue codes over
$\mathbb{Z}_4$,  IEEE Trans. Inform. Theory, {\bf 42}(1996),  1594-1600.


\bibitem{Pless2}
V. Pless, P. Sol\'{e}, Z. Qian, Cyclic self-dual $\mathbb{Z}_4$-codes, Finite Fields
Appl., {\bf3}(1997), 48-69.


\bibitem{Sloane}
N. J. A. Sloane, J. G. Thompson, Cyclic self-dual codes, IEEE Trans. Inform. Theory, {\bf 29}(1983),
364-367.


\bibitem{Wan99}
Z. Wan, Cyclic codes over Galois rings,  Alg. Colloq., {\bf 6}(1999),
291-304.


\bibitem{Wan}
Z.  Wan, Lectures on Finite Fields and Galois Rings, World Scientific Publishing, Singapore,  2003.



\bibitem{Wood99}
J. Wood,   Duality for modules over finite rings and applications to coding theory,  Am. J. Math., {\bf 121}(1999)
555-575.

\bibitem{Wood08}
J. Wood, Code equivalence characterizes finite Frobenius rings,  Proc. Am. Math. Soc., {\bf 136}(2008),  699-706.


\end{thebibliography}
\end{document}